\documentclass[final,3p,12pt,sort&compress]{elsarticle}
\makeatletter
\def\ps@pprintTitle{%
 \let\@oddhead\@empty
 \let\@evenhead\@empty
 \def\@oddfoot{\centerline{\thepage}}%
 \let\@evenfoot\@oddfoot}
\makeatother
\usepackage{fullpage}
\usepackage{indentfirst}
\usepackage{verbatim}
\usepackage{amssymb,amsmath, amsthm}
\usepackage{slashed}
\usepackage[bookmarks]{hyperref}
\numberwithin{equation}{section}
\theoremstyle{plain}% default
\newtheorem{thm}{Theorem}[section]

\newtheorem{prop}[thm]{Proposition}
\newtheorem{rem}[thm]{Remark}

\theoremstyle{definition}
\newtheorem{defn}[thm]{Definition}

\theoremstyle{remark}

\journal{Nuclear Physics B}

\begin{document}

\begin{frontmatter}


\title{The Majorana spinor representation of the Poincare group}
%%\author{Leonardo Pedro\corref{cor1}\fnref{label2}} 

%\title{}

\author[label1]{Leonardo Pedro}
\address[label1]{Centro de Fisica Teorica de Particulas, CFTP,
  Departamento de Fisica, Instituto Superior Tecnico, Universidade
  Tecnica de Lisboa, Avenida Rovisco Pais nr. 1, 1049-001 Lisboa, Portugal}
\ead{leonardo@cftp.ist.utl.pt}
\date{\today}

%\author{}

%\address{}

\begin{abstract}

There are Poincare group representations
on complex Hilbert spaces, like the Dirac spinor field, 
or real Hilbert spaces, like the electromagnetic field tensor.
The Majorana spinor is an element of a 4
dimensional real vector space. The Majorana spinor field is
a space-time dependent Majorana spinor, solution of the free Dirac equation.

The Majorana-Fourier
and Majorana-Hankel transforms of Majorana spinor fields are defined
and related to the linear and angular momenta of a spin one-half
representation of the Poincare group.
We show that the Majorana spinor field with finite mass is an unitary
irreducible projective representation of the Poincare group on a real
Hilbert space.

Since the Bargmann-Wigner
equations are valid for all spins and are based on
the free Dirac equation, these results
open the possibility to study Poincare group representations with
arbitrary spins on real Hilbert spaces.

\end{abstract}


\end{frontmatter}

%\tableofcontents
%\pagebreak
\section{Introduction}

The irreducibility of a group representation may
depend on whether the representation space is a real or complex
Hilbert space. There are Poincare group representations
on complex Hilbert spaces, like the Dirac spinor fields, or real
Hilbert spaces, like the electromagnetic field tensor.

The Poincare group, also called inhomogeneous Lorentz group, is the
semi-direct product of the translations and
Lorentz groups\cite{Hall}.  Whether or not the Lorentz and Poincare groups
include the parity and time reversal transformations depends on the context and
authors. To be clear, we use the prefixes full/restricted when
including/excluding parity and time reversal transformations. A
projective representation of the Poincare group on a complex/real Hilbert space is an homomorphism,
defined up to a complex phase/sign, from the group to the automorphisms of the Hilbert
space. The Pin(3,1) group representations are projective representations of the full Lorentz
group\cite{pin}, while the SL(2,C) subgroup representations are projective
representations of the restricted Lorentz subgroup.

The unitary projective representations of the Poincare group on
complex Hilbert spaces were studied by many authors, including
Wigner\cite{knapp,wigner,mackey,ohnuki,poincare,weinberg}.
Since Quantum Mechanics is based on
complex Hilbert spaces \cite{yang}, these studies were very
important in the evolution of the role of symmetry in the Quantum
theory\cite{symmetry}. Although Quantum Theory in real Hilbert spaces
was investigated before 
\cite{realQM, realqft, realqftII, realqftIII, quantumstatistics,
  hestenes_old, hestenes_recent}, to our knowledge, the unitary
projective representations of the Poincare group on real Hilbert
spaces were not studied.

The Dirac spinor is an element of a 4 dimensional
complex vector space, while the Majorana spinor is an
element of a 4 dimensional real vector space 
\cite{todorov}. The Majorana spinor
representation of both SL(2,C) and Pin(3,1) is
irreducible \cite{irreducible}. The spinor fields,
space-time dependent spinors, are solutions of the free Dirac equation
\cite{Dirac}. The Hilbert space of Dirac spinor fields is complex,
while the Hilbert space of Majorana spinor fields is real.

To study a system of many neutral particles with spin one-half,
Majorana spinor fields are extended with second quantization operators
and are called Majorana quantum fields or Majorana fermions 
\cite{Majorana, pal, dreiner}. There are important
applications of the Majorana quantum field in
theories trying to explain phenomena in neutrino physics, dark
matter searches, the fractional quantum Hall effect and superconductivity
\cite{solidstate}. Note that Majorana quantum fields are related to but
are different from the Majorana spinor fields.

The Bargmann-Wigner equations\cite{BW,allspins} are based on the free
Dirac equation and are valid for all spins.
The free Dirac equation is diagonal in the Newton-Wigner
representation\cite{newton}, related to the Dirac representation
through a Foldy-Wouthuysen transformation \cite{revfoldy,foldy}.
In the context of Clifford Algebras, there are studies
on the geometric square roots of -1 
\cite{hestenes_old, hestenes_recent, squareroot} and on the
generalizations of the Fourier transform \cite{clifford}, with
applications to image processing\cite{image}.

In the following we will study the spin one-half representation of the
Poincare group on the real Hilbert space of Majorana spinor fields.
In chapter 2 we define the Majorana matrices and spinors. 
In chapter 3 we study the Majorana spinor projective representation of
the Lorentz group and show that the Majorana spinor representations of the groups
SU(2), SL(2,C) and Pin(3,1) are irreducible.
In chapter 4 we relate the Majorana and Pauli spinor fields.
In 5 and 6 we define the Majorana-Fourier and Majorana-Hankel
transforms of a Majorana spinor.
In 7 we show that the projective Poincare group representation on the Majorana
spinor field is unitary and irreducible. 
We relate the Majorana transforms to the linear and angular momenta of
a spin one-half representation of the Poincare group and show that the transition
operator is causal. In 8, we extend the Majorana transforms to include
the energy.

\section{Majorana, Dirac and Pauli Matrices and Spinors}
The Majorana matrices, $i\gamma^\mu$ with $\mu=0,1,2,3$,
are the Dirac Gamma matrices, $\gamma^\mu$, times the imaginary unit.
The notation maintains explicit the relation between the Majorana and
Dirac Gamma matrices. 

\begin{defn}
The Majorana matrices, $i\gamma^\mu$, are $4\times 4$ unitary matrices
with anti-commutator $\{i\gamma^\mu,i\gamma^\nu\}$:
\begin{align*}
(i\gamma^\mu)(i\gamma^\nu)+(i\gamma^\nu)(i\gamma^\mu)=-2g^{\mu\nu},\ \mu,\nu=0,1,2,3
\end{align*}
Where $g=diag(1,-1,-1,-1)$ is the Minkowski metric.
The pseudo-scalar is $i\gamma^5\equiv
-\gamma^0\gamma^1\gamma^2\gamma^3$.
\end{defn}

\begin{rem}
Pauli's fundamental theorem\cite{diracmatrices} implies that the Majorana matrices are unique up to
an unitary similarity transformation.
\end{rem}

The product of 2 Dirac Gamma
matrices is minus the
product of 2 corresponding Majorana matrices:
$\gamma^\mu\gamma^\nu=-i\gamma^\mu i\gamma^\nu$.

In a Majorana basis, the Majorana matrices are $4\times 4$ real
orthogonal matrices. An example of the Majorana matrices in a particular Majorana basis is:
\begin{align*}
\begin{array}{llllll}
\label{basis}
i\gamma^1=&\left[ \begin{smallmatrix}
+1 & 0 & 0 & 0 \\
0 & -1 & 0 & 0 \\
0 & 0 & -1 & 0 \\
0 & 0 & 0 & +1 \end{smallmatrix} \right]&
i\gamma^2=&\left[ \begin{smallmatrix}
0 & 0 & +1 & 0 \\
0 & 0 & 0 & +1 \\
+1 & 0 & 0 & 0 \\
0 & +1 & 0 & 0 \end{smallmatrix} \right]&
i\gamma^3=\left[ \begin{smallmatrix}
0 & +1 & 0 & 0 \\
+1 & 0 & 0 & 0 \\
0 & 0 & 0 & -1 \\
0 & 0 & -1 & 0 \end{smallmatrix} \right]\\
\\
i\gamma^0=&\left[ \begin{smallmatrix}
0 & 0 & +1 & 0 \\
0 & 0 & 0 & +1 \\
-1 & 0 & 0 & 0 \\
0 & -1 & 0 & 0 \end{smallmatrix} \right]&
i\gamma^5=&\left[ \begin{smallmatrix}
0 & -1 & 0 & 0 \\
+1 & 0 & 0 & 0 \\
0 & 0 & 0 & +1 \\
0 & 0 & -1 & 0 \end{smallmatrix} \right]&
=-\gamma^0\gamma^1\gamma^2\gamma^3
\end{array}
\end{align*}

In reference \cite{realgamma} it is proved that the set of five
anti-commuting $4\times 4$ real matrices is unique up to
isomorphisms. So it is not possible to obtain the euclidean signature
for the metric, for instance.

\begin{defn}
The Dirac spinor is a $4\times 1$ complex column matrix, that
transforms in a precise way under the action of Lorentz
transformations.
\end{defn}

The space of Dirac spinors is a 4 dimensional complex vector space.

\begin{defn}
Let $S$ be a unitary matrix such that $S i\gamma^\mu S^{\dagger}$ is
real, for $\mu=0,1,2,3$.

The set of Majorana spinors, $Pinor$, is the subset of Dirac spinors $u$
verifying the Majorana condition:
\begin{align*}
(S u)^*=(S u)
\end{align*}
Where $^*$ denotes complex conjugation and $^\dagger$ denotes
hermitian conjugate.
\end{defn}

\begin{rem}
Let $W$ be a subset of a vector space $V$ over $\mathbb{C}$. 
$W$ is a real vector space iff:
 $0\in W$;
 If $u,v\in W$, then $u+v\in W$;
 If $u\in W$ and $c\in \mathbb{R}$, then $c u\in W$.
\end{rem}

From the previous remark, the set of Majorana spinors is a 4
dimensional real vector space, while the set of Dirac spinors is a 8
dimensional real vector space. Note that the linear combinations of
Majorana spinors with complex scalars do not verify the Majorana
condition. The Majorana spinor, in a Majorana basis, is a $4\times 1$
real column matrix.

\begin{defn}
The Pauli matrices $\sigma^k,\ k\in\{1,2,3\}$ are $2\times 2$
hermitian, unitary, anti-commuting, complex matrices.
The Pauli spinor is a  $2\times 1$ complex column matrix. The space of
Pauli spinors is denoted by $Pauli$.
\end{defn}

The space of Pauli spinors, $Pauli$, is a 2 dimensional complex vector
space and a 4 dimensional real vector space.

\begin{rem}
Pauli's fundamental theorem guarantees that the Pauli matrices are unique up to
an unitary similarity transformation.
\end{rem}

\section{Majorana spinor representation of the Lorentz group}

\begin{rem} 
The Lorentz group, $O(1,3)\equiv\{\lambda \in \mathbb{R}^{4\times 4}: \lambda^T \eta \lambda=\eta \}$, is the set of
real matrices that leave the metric, $\eta=diag(1,-1,-1,-1)$,
invariant.

The proper orthochronous Lorentz subgroup is defined by $SO^+(1,3)\equiv\{\lambda \in
O(1,3): det(\lambda)=1, \lambda^0_{\ 0}>0 \}$. 
It is a normal subgroup. The discrete Lorentz subgroup of parity and time-reversal is $\Delta \equiv \{1,\eta,-\eta,-1\}$.

The Lorentz group is the semi-direct product of the previous
subgroups, $O(1,3)=\Delta \ltimes SO^+(1,3)$.  
\end{rem}

\begin{rem}
$Pin(3,1)$ \cite{pin} is the group of endomorphisms of Majorana
spinors that leave the space of linear
combinations of the Majorana matrices invariant, that is:
\begin{align*}
Pin(3,1)\equiv \Big\{S\in End(Pinor):\ det S=1,\
S^{-1}(i\gamma^\mu)S=\Lambda^\mu_{\ \nu}i\gamma^\nu,\
 \Lambda\in O(1,3) \Big\}
\end{align*}

The map $\Lambda:Pin(3,1)\to O(1,3)$ defined by:
\begin{align*}
(\Lambda(S))^\mu_{\ \nu}i\gamma^\nu\equiv S^{-1}(i\gamma^\mu)S
\end{align*}
is two-to-one and surjective. It defines a group homomorphism.

$Pin(3,1)$ is the semi-direct product of the groups 
$Spin^+(3,1)\equiv \{e^{\theta^j
  i\gamma^5\gamma^0\gamma^j+b^j\gamma^0\gamma^j}: \theta^j,b^j\in
\mathbb{R},\ j\in\{1,2,3\}\}$
and $\Omega \equiv \{\pm 1,\pm i\gamma^0,\pm \gamma^0\gamma^5,\pm
i\gamma^5\}$. The group homomorphisms 
$\Lambda:Spin^+(3,1)\to SO^+(1,3)$ and $\Lambda:\Omega\to \Delta$ are
two-to-one and surjective. $Spin^+(3,1)$ is isomorphic to
$SL(2,\mathbb{C})$, while the unitary subgroup $Spin^+(3,1)\cap SU(4)=\{e^{\theta^j
  i\gamma^5\gamma^0\gamma^j}: \theta^j\in
\mathbb{R},\ j\in\{1,2,3\}\}$
is isomorphic to $SU(2)$.
\end{rem}

\begin{defn}
The Majorana spinor representation of $Pin(3,1)$ and subgroups
is defined by the action of $S\in Pin(3,1)$ in the space of
Majorana spinors.
\end{defn}

\begin{rem}
A unitary matrix representation of a group is irreducible iff there is no
basis where all the matrices of the representation can be block
diagonalized (in a non-trivial way).
\end{rem}

\begin{prop}
\label{prop:irreducible}
The Majorana spinor representation of
$Spin^+(1,3)\cap SU(4)$  (isomorphic to $SU(2)$), is irreducible.
\end{prop}

\begin{proof}
In a Majorana basis, the automorphisms of
Majorana spinors are $4\times 4$ non-singular real matrices. 
We can check that $i\gamma^5\gamma^0\gamma^j\in  Spin^+(1,3)\cap
SU(4)$, $j\in\{1,2,3\}$. 
These matrices square to $-1$ and anti-commute. If there is a basis
where they are all block diagonal, then the blocks also square to $-1$
and anti-commute. But there is only one (linear independent) $2\times
2$ real matrix 
that squares to $-1$ and no $1\times 1$ real matrix that
squares to $-1$. Therefore, the
representation is irreducible.
\end{proof}

\section{Hilbert spaces of Majorana and Pauli spinor fields}
\begin{defn}
The complex Hilbert space of Pauli spinors, $Pauli$, has the internal
product:
\begin{align*}
<\phi,\psi>=\phi^\dagger\psi;\ \phi,\psi\in Pauli
\end{align*}
\end{defn}

\begin{defn}
The real Hilbert space of Majorana spinors, $Pinor$, has the internal
product:
\begin{align*}
<\Phi,\Psi>=\Phi^\dagger\Psi;\ \Phi,\Psi\in Pinor
\end{align*}
\end{defn}

\begin{defn}
\label{defn:Theta}
Consider that $\{M_+,M_-,i\gamma^0M_+,i\gamma^0M_-\}$ and $\{P_+,P_-,iP_+,iP_-\}$ are orthonormal basis of
the 4 dimensional real vector spaces $Pinor$ and $Pauli$, respectively, verifying:
\begin{align*}
\gamma^3\gamma^5 M_\pm=\pm M_\pm&,\ \sigma^3 P_\pm=\pm P_\pm
\end{align*}
Let $H$ be a real Hilbert space. 
For all $h\in H$, the bijective linear map
$\Theta_H:Pauli\otimes_{\mathbb{R}} H\to Pinor\otimes_{\mathbb{R}}H$ is defined by:
\begin{align*}
\Theta_H(h\otimes_{\mathbb{R}} P_+)=h\otimes_{\mathbb{R}} M_+,&\ \Theta_H(h \otimes_{\mathbb{R}}
iP_+)=h\otimes_{\mathbb{R}} i\gamma^0 M_+\\
\Theta_H(h\otimes_{\mathbb{R}} P_-)=h\otimes_{\mathbb{R}} M_-,&\ \Theta_H(h\otimes_{\mathbb{R}} iP_-)=h\otimes_{\mathbb{R}}i\gamma^0 M_-
\end{align*}
\end{defn}

\begin{defn}
Let $H_n$, with $n\in\{1,2\}$, be two real Hilbert spaces 
and $U:Pauli\otimes_{\mathbb{R}} H_1\to
Pauli\otimes_{\mathbb{R}} H_2$ be an operator.
The operator $U^\Theta:Pinor\otimes_{\mathbb{R}} H_1\to
Pinor\otimes_{\mathbb{R}} H_2$ is defined as
$U^\Theta\equiv
\Theta_{H_2}\circ U\circ \Theta^{-1}_{H_1}$.
\end{defn}

\begin{rem}
Let $H_n$, with $n\in\{1,2\}$, be two Hilbert spaces with internal
products $<,>:H_n\times H_n\to
\mathbb{F}$,($\mathbb{F}=\mathbb{R},\mathbb{C}$). 
A linear operator $U:H_1\to H_2$ is unitary iff:

1) it is surjective;

2) for all $x\in H_1$, $<U(x) , U(x)>=<x, x>$.
\end{rem}

\begin{rem}
Given two real Hilbert spaces $H_1$, $H_2$ and an unitary operator
$U:H_1\to H_2$, the inverse operator $U^{-1}:H_2\to H_1$ is defined by:
\begin{align*}
<x,U^{-1} y>=<U x,y>,\ x\in H_1, y\in H_2
\end{align*}
\end{rem}

\begin{prop}
\label{prop:unitary}
Let $H_n$, with $n\in\{1,2\}$, be two real Hilbert spaces.
The following two statements are equivalent:

1) The operator $U:Pauli\otimes_{\mathbb{R}} H_1\to
Pauli\otimes_{\mathbb{R}} H_2$ is unitary;

2) The operator $U^\Theta:Pinor\otimes_{\mathbb{R}} H_1\to
Pinor\otimes_{\mathbb{R}} H_2$ is unitary.
\end{prop}
\begin{proof}
Because $\Theta_{H_n}$ is bijective, $U$ is surjective iff
$\Theta_{H_2}\circ U\circ \Theta^{-1}_{H_1}$ is surjective.

For all $g\in Pauli\otimes_{\mathbb{R}} H_1$, we have:
\begin{align*}
<g,g>=<\Theta_{H_1}(g),\Theta_{H_1}(g)>
\end{align*}
\begin{align*}
<U(g),U(g)>=<\Theta_{H_2}(U(g)),\Theta_{H_2}(U(g))>
\end{align*}

Since $\Theta_{H_n}$ is bijective, we get that the following two statements are
equivalent:

1) for all $g\in Pauli\otimes_{\mathbb{R}} H_1$, $<g,g>=<U(g),U(g)>$;

2) for all $g'\in Pinor\otimes_{\mathbb{R}}
H_1$, $<g',g'>=<\Theta_{H_2}(U(\Theta^{-1}_{H_1}(g'))),\Theta_{H_2}(U(\Theta^{-1}_{H_1}(g')))>$.
\end{proof}

\begin{defn}
The space of Majorana spinor fields over a set $S$, $Pinor(S)\equiv
Pinor\otimes_{\mathbb{R}} L^2(S)$, is the
real Hilbert space of Majorana spinors whose
entries, in a Majorana basis, are real Lebesgue square integrable functions of
$S$.
\end{defn}

\begin{defn}
The space of Pauli spinor fields over a set $S$,
$Pauli(S)\equiv Pauli \otimes_{\mathbb{R}}
L^2(S)$ is the complex Hilbert space of Pauli spinors whose
components are complex Lebesgue square integrable functions of
$S$.
\end{defn}

\section{Linear Momentum of Majorana spinor fields}
\begin{defn}
$L^2(\mathbb{R}^n)$ is the real Hilbert space of real functions of $n$ real
variables whose square is Lebesgue integrable in $\mathbb{R}^n$.
The internal product is:
\begin{align*}
<f,g>\equiv\int d^nx f(x)g(x),\ f,g \in L^2(\mathbb{R}^n)
\end{align*} 
\end{defn}

\begin{rem}
The Pauli-Fourier Transform $\mathcal{F}_P: Pauli(\mathbb{R}^n)\to
Pauli(\mathbb{R}^n)$ is an unitary operator defined by:
\begin{align*}
\mathcal{F}_P\{\psi\}(\vec{p})\equiv\int d^n\vec{x} \frac{e^{-i\vec{p}\cdot
  \vec{x}}}{\sqrt{(2\pi)^n}}\psi(\vec{x}),\ \psi\in Pauli(\mathbb{R}^n)
\end{align*}
Where the domain of the integral is $\mathbb{R}^n$.
\end{rem}

\begin{defn}
The Majorana-Fourier Transform $\mathcal{F}_M: Pinor(\mathbb{R}^3)\to
Pinor(\mathbb{R}^3)$ is an operator defined by:
\begin{align*}
\mathcal{F}_M\{\Psi\}(\vec{p})&\equiv \int d^3\vec{x}\ \frac{e^{-i\gamma^0\vec{p}\cdot
  \vec{x}}}{\sqrt{(2\pi)^3}}
\frac{\slashed p \gamma^0+m}{\sqrt{E_p+m}\sqrt{2E_p}}\Psi(\vec{x}),
\ \Psi\in Pinor(\mathbb{R}^3)
\end{align*}
Where the domain of the integral is $\mathbb{R}^3$, $m\geq 0$,
$E_p\equiv\sqrt{\vec{p}^2+m^2}$ and $\slashed p=
E_p\gamma^0-\vec{p}\cdot \vec{\gamma}$.
\end{defn}

\begin{prop}
The Majorana-Fourier Transform is an unitary operator.
\end{prop}

\begin{proof}
The Majorana-Fourier Transform can be written as:
\begin{align*}
\mathcal{F}_M\{\Psi\}(\vec{p})\equiv& \sqrt{\frac{E_p+m}{2E_p}}\Big(\int d^3\vec{x}\ \frac{e^{-i\gamma^0\vec{p}\cdot
  \vec{x}}}{\sqrt{(2\pi)^3}}\Psi(\vec{x})\Big)\\
-&\sqrt{\frac{E_p-m}{2E_p}}\frac{\vec{p}\cdot\vec{\gamma}\gamma^0}{|\vec{p}|}\Big(\int d^3\vec{x}\ \frac{e^{+i\gamma^0\vec{p}\cdot
  \vec{x}}}{\sqrt{(2\pi)^3}}\Psi(\vec{x})\Big)
\end{align*}
So, one gets:
\begin{align*}
\mathcal{F}_{M}\{\Psi\}=S \circ \mathcal{F}^\Theta_{P}\{\Psi\}
\end{align*}
Where $S:Pinor(\mathbb{R}^3)\to Pinor(\mathbb{R}^3)$ is a bijective linear map defined by:
\begin{align*}
\left[ \begin{array}{l}
S\{\Psi\}(+\vec{p})\\
S\{\Psi\}(-\vec{p})
\end{array}
\right]&\equiv
\left[ \begin{array}{cc}
\sqrt{\frac{E_p+m}{2E_p}} 
& -\sqrt{\frac{E_p-m}{2E_p}}\frac{\vec{p}\cdot\vec{\gamma}\gamma^0}{|\vec{p}|}\\
\sqrt{\frac{E_p-m}{2E_p}}\frac{\vec{p}\cdot\vec{\gamma}\gamma^0}{|\vec{p}|}
&\sqrt{\frac{E_p+m}{2E_p}}
\end{array}
\right]\ 
\left[ \begin{array}{l}
\Psi(+\vec{p})\\
\Psi(-\vec{p})
\end{array}
\right]
\end{align*}
We can check that the $2\times 2$ matrix appearing in the equation
above is orthogonal. Therefore $S$ is an unitary operator.
Since $\mathcal{F}^\Theta_{P}$ is also unitary,
$\mathcal{F}_{M}$ is unitary.
\end{proof}

\begin{prop}
The inverse Majorana-Fourier Transform verifies:
\begin{align*}
(\gamma^0\vec{\gamma}\cdot \vec{\partial}+i\gamma^0
m)\mathcal{F}_M^{-1}\{\Psi\}(\vec{x})&=(\mathcal{F}_M^{-1}\circ
R)\{\Psi\}(\vec{x})\\
\vec{\partial}_j\mathcal{F}_M^{-1}\{\Psi\}(\vec{x})&=(\mathcal{F}_M^{-1}\circ
R_j)\{\Psi\}(\vec{x})
\end{align*}
Where $\Psi\in Pinor(\mathbb{R}^3)$ and $R,R_j:Pinor(\mathbb{R}^3)\to Pinor(\mathbb{R}^3)$ are
 linear maps defined by $R\{\Psi\}(\vec{p})=i\gamma^0
 E_p\Psi(\vec{p})$ and 
$R_j\{\Psi\}(\vec{p})=i\gamma^0 \vec{p}_j\Psi(\vec{p})$ .
\end{prop}

\begin{proof}
We have $\mathcal{F}^{-1}_{M}=(\mathcal{F}^{\Theta}_{P})^{-1}\circ S^{-1}$. Then:
\begin{align*}
(\gamma^0\vec{\gamma}\cdot \vec{\partial}+i\gamma^0
m)(\mathcal{F}^\Theta_P)^{-1}\{\Psi\}(\vec{x})=((\mathcal{F}^\Theta_P)^{-1}\circ Q) \{\Psi\}(\vec{x})
\end{align*}
Where $Q:Pinor(\mathbb{R}^3)\to Pinor(\mathbb{R}^3)$ is a linear map defined by:
\begin{align*}
\left[ \begin{array}{l}
Q\{\Psi\}(+\vec{p})\\
Q\{\Psi\}(-\vec{p})
\end{array}
\right]&\equiv
\left[ \begin{array}{cc}
i\gamma^0 m 
& i\vec{p}\cdot \vec{\gamma}\\
-i\vec{p}\cdot \vec{\gamma}
& i\gamma^0 m
\end{array}
\right]\ 
\left[ \begin{array}{l}
\Psi(+\vec{p})\\
\Psi(-\vec{p})
\end{array}
\right]
\end{align*}
Now we show that $Q\circ S^{-1}=S^{-1}\circ R$:
\begin{align*}
&\left[ \begin{array}{cc}
i\gamma^0 m 
& i\vec{p}\cdot \vec{\gamma}\\
-i\vec{p}\cdot \vec{\gamma}
& i\gamma^0 m
\end{array}
\right]\ 
\left[ \begin{array}{cc}
\sqrt{\frac{E_p+m}{2E_p}} 
& \sqrt{\frac{E_p-m}{2E_p}}\frac{\vec{p}\cdot\vec{\gamma}\gamma^0}{|\vec{p}|}\\
-\sqrt{\frac{E_p-m}{2E_p}}\frac{\vec{p}\cdot\vec{\gamma}\gamma^0}{|\vec{p}|}
&\sqrt{\frac{E_p+m}{2E_p}}
\end{array}
\right]=\\
&=\left[ \begin{array}{cc}
\sqrt{\frac{E_p+m}{2E_p}} 
& \sqrt{\frac{E_p-m}{2E_p}}\frac{\vec{p}\cdot\vec{\gamma}\gamma^0}{|\vec{p}|}\\
-\sqrt{\frac{E_p-m}{2E_p}}\frac{\vec{p}\cdot\vec{\gamma}\gamma^0}{|\vec{p}|}
&\sqrt{\frac{E_p+m}{2E_p}}
\end{array}\right]\ 
 \left[ \begin{array}{cc}
 i\gamma^0 E_p & 0\\
 0 & i\gamma^0 E_p
 \end{array}
 \right]
\end{align*}

We also have that:
\begin{align*}
\vec{\partial}_j(\mathcal{F}^\Theta_P)^{-1}\{\Psi\}(\vec{x})=((\mathcal{F}^\Theta_P)^{-1}\circ R_j) \{\Psi\}(\vec{x})
\end{align*}
Where $R_j:Pinor(\mathbb{R}^3)\to Pinor(\mathbb{R}^3)$ is the linear map defined by:
\begin{align*}
\left[ \begin{array}{l}
R_j\{\Psi\}(+\vec{p})\\
R_j\{\Psi\}(-\vec{p})
\end{array}
\right]&\equiv
\left[ \begin{array}{cc}
i\gamma^0 \vec{p}_j 
& 0\\
0
& -i\gamma^0 \vec{p}_j 
\end{array}
\right]\ 
\left[ \begin{array}{l}
\Psi(+\vec{p})\\
\Psi(-\vec{p})
\end{array}
\right]
\end{align*}
It verifies $R_j\circ S^{-1}=S^{-1}\circ R_j$.
\end{proof}

\section{Angular momentum of Majorana spinor fields}
\begin{defn}
Let $\vec{x}\in \mathbb{R}^3$. The spherical coordinates
parametrization is:
\begin{align*}
\vec{x}=r(\sin(\theta)\sin(\varphi)\vec{e_1}+\sin(\theta)\sin(\varphi)\vec{e_2}+\cos(\theta)\vec{e}_3)
\end{align*}
where $\{\vec{e}_1,\vec{e}_2,\vec{e}_3\}$ is a fixed orthonormal basis of
$\mathbb{R}^3$ and $r\in [0,+\infty[$, $\theta \in [0,\pi]$, $\varphi
\in [-\pi,\pi]$.
\end{defn}

\begin{defn}
Let
\begin{align*}
\mathbb{S}^3\equiv \{(p,l,\mu):p\in \mathbb{R}_{\geq 0}; 
l,\mu \in \mathbb{Z}; l\geq 1; -l \leq \mu\leq l-1 \}
\end{align*}
The Hilbert space $L^2(\mathbb{S}^3)$ is the real Hilbert space of real
Lebesgue square integrable functions of $\mathbb{S}^3$. The internal product is:
\begin{align*}
<f,g>=\sum_{l=1}^{+\infty}\sum_{\mu=-l}^{l-1}\int_0^{+\infty} dp f(p,l,\mu)
g(p,l,\mu),\ f,g\in L^2(\mathbb{S}^3)
\end{align*}
\end{defn}

\begin{defn}
The Pauli-Hankel transform $\mathcal{H}_{P}: Pauli(\mathbb{R}^3)\to
Pauli(\mathbb{S}^3)$ is an operator defined by:
\begin{align*}
\mathcal{H}_{P}\{\psi\}(p,l,\mu)\equiv\int r^2 dr d(\cos\theta)d\varphi
\frac{2 p}{\sqrt{2\pi}}\lambda^\dagger_{l \mu}(pr,\theta,\varphi) \psi(r,\theta,\varphi),\ \psi\in Pauli(\mathbb{R}^3)
\end{align*}
The domain of the integral is $\mathbb{R}^3$. The matrices $\lambda_{l
  \mu}$, the spherical
Bessel function of the first kind $j_n$ \cite{bessel}, the  Pauli
spherical matrices $\omega_{l \mu}$\cite{harmonics}, 
the spherical harmonics $Y_{l\mu}$ and the associated Legendre
functions of the first kind $P_{l\mu}$ are:
\begin{align*}
\lambda_{l \mu}(r,\theta,\varphi)
\equiv& \omega_{l \mu}(\theta,\varphi)\Big(j_l(r)\frac{1+\sigma^3}{2}+j_{l-1}(r)\frac{1-\sigma^3}{2}\Big)\\
j_l(r)\equiv& r^l\Big(-\frac{1}{r}\frac{d}{dr}\Big)^l \frac{\sin
  r}{r}\\
\omega_{l \mu}(\theta,\varphi)
\equiv& \Big(-\sqrt{\frac{l-\mu}{2l+1}}
Y_{l,\mu}(\theta,\varphi)+\sqrt{\frac{l+\mu+1}{2l+1}}
Y_{l,\mu+1}(\theta,\varphi)\sigma^1\Big)\frac{1+\sigma^3}{2}\\
+&\Big(\sqrt{\frac{l+\mu}{2l-1}}
Y_{l-1,\mu}(\theta,\varphi)\sigma^1+\sqrt{\frac{l-\mu-1}{2l-1}}
Y_{l-1,\mu+1}(\theta,\varphi)\Big)\frac{1-\sigma^3}{2}\\
Y_{l\mu}(\theta,\varphi)\equiv&\sqrt{\frac{2l+1}{4\pi}\frac{(l-m)!}{(l+m)!}}
P_{l}^\mu(\cos\theta)e^{i\mu \varphi}\\
P_{l}^\mu(\xi)\equiv&\frac{(-1)^{\mu}}{2^{l}l!}(1-\xi^{2})^{\mu/2}
\frac{\mathrm{d}^{l+\mu}}{\mathrm{d}\xi^{l+\mu}}(\xi^{2}-1)^{l}
\end{align*}
\end{defn}

\begin{rem}
Due to the properties of spherical harmonics and Bessel functions, the
Pauli-Hankel transform is an unitary operator. The inverse Pauli-Hankel Transform verifies:
\begin{align*}
\vec{\sigma}\cdot \vec{\partial}\
\mathcal{H}_P^{-1}\{\psi\}(\vec{x})&=(\mathcal{H}_P^{-1}\circ
R)\{\psi\}(\vec{x})\\
(\frac{1}{2}\sigma^3-x^1i\partial_2+x^2i\partial_1)\ 
\mathcal{H}_P^{-1}\{\psi\}(\vec{x})&=(\mathcal{H}_P^{-1}\circ
R')\{\psi\}(\vec{x})
\end{align*}
Where  $\psi\in Pauli(\mathbb{S}^3)$  and $R,R':Pauli(\mathbb{S}^3)\to Pauli(\mathbb{S}^3)$ are
linear maps defined by:
\begin{align*}
R\{\psi\}(p,l,\mu)&\equiv p \sigma^1\sigma^3\psi(p,l,\mu)\\
R'\{\psi\}(p,l,\mu)&\equiv (\mu+\frac{1}{2})\psi(p,l,\mu)
\end{align*}
\end{rem}

\begin{defn}
The Majorana-Hankel transform $\mathcal{H}_{M}: Pinor(\mathbb{R}^3)\to
Pinor(\mathbb{S}^3)$ is an operator defined by:
\begin{align*}
\mathcal{H}_{M}\{\Psi\}(p,l,\mu)&\equiv\int r^2 dr
d(\cos\theta)d\varphi
\frac{2 p}{\sqrt{2\pi}}\Delta^\dagger(p,l,\mu,r,\theta,\varphi)\Psi(r,\theta,\varphi),\
\Psi\in Pinor(\mathbb{R}^3)
\end{align*}
\begin{align*}
\Delta(p,l,\mu,r,\theta,\varphi)\equiv \sqrt{\frac{E_p+m}{2E_p}}\Lambda_{l\mu}(pr,\theta,\varphi)
+\sqrt{\frac{E_p-m}{2E_p}}(-1)^\mu\Lambda_{l,-\mu-1}(pr,\theta,\varphi)i\gamma^3
\end{align*}
%\Big(\sqrt{\frac{E_p+m}{2E_p}}
%j_l(pr)+\sqrt{\frac{E_p-m}{2E_p}}j_{l-1}(pr)i\gamma^r\Big)
%\Omega_{l \mu}(\theta,\varphi)\frac{1+\gamma^3\gamma^5}{2}\\
%+&\Big(\sqrt{\frac{E_p+m}{2E_p}} j_{l-1}(pr)-\sqrt{\frac{E_p-m}{2E_p}}  j_{l}(pr)i\gamma^r\Big)\Omega_{l
%  \mu}(\theta,\varphi)\frac{1-\gamma^3\gamma^5}{2}
Where the matrices 
$\Lambda_{l\mu}(r,\theta,\varphi)\equiv \Theta\circ\lambda_{l\mu}(r,\theta,\varphi)\circ \Theta^{-1}$ are obtained from 
the Pauli matrices $\lambda_{l\mu}$ replacing $(i,\sigma^1,\sigma^3)$ by $(i\gamma^0,\gamma^1\gamma^5,\gamma^3\gamma^5)$.
\end{defn}

\begin{prop}
The Majorana-Hankel transform is an unitary operator.
\end{prop}
\begin{proof}
%Using the property  $\sigma^r\omega_{l \mu}=-\omega_{l
%  \mu}\sigma^1$ \cite{harmonics}, we get that $i\gamma^r\Omega_{l
%  \mu}=(-1)^\mu\Omega_{l,-\mu-1}i\gamma^5$.

The Majorana-Hankel transform can be written as:
\begin{align*}
\mathcal{H}_{M}=S \circ \mathcal{H}^\Theta_{P}
\end{align*}
Where $S:Pinor(\mathbb{S}^3)\to Pinor(\mathbb{S}^3)$ is a bijective linear map defined by:
\begin{align*}
\left[ \begin{array}{l}
S\{\Psi\}(p,l,\mu)\\
S\{\Psi\}(p,l,-\mu-1)
\end{array}
\right]&\equiv
\left[ \begin{array}{cc}
\sqrt{\frac{E_p+m}{2E_p}} 
& \sqrt{\frac{E_p-m}{2E_p}}(-1)^\mu i\gamma^3\\
-\sqrt{\frac{E_p-m}{2E_p}}(-1)^\mu i\gamma^3 
&\sqrt{\frac{E_p+m}{2E_p}}
\end{array}
\right]\ 
\left[ \begin{array}{l}
\Psi(p,l,\mu)\\
\Psi(p,l,-\mu-1)
\end{array}
\right]
\end{align*}
We can check that the $2\times 2$ matrix appearing in the equation
above is orthogonal. Therefore $S$ is an unitary operator.
Since $\mathcal{H}^\Theta_{P}$ is also unitary,
$\mathcal{H}_{M}$ is unitary. 
\end{proof}

\begin{prop}
The inverse Majorana-Hankel Transform verifies:
\begin{align*}
(\gamma^0\vec{\gamma}\cdot \vec{\partial}+i\gamma^0 m)
\mathcal{H}_M^{-1}\{\Psi\}(\vec{x})&=(\mathcal{H}_M^{-1}\circ
R)\{\Psi\}(\vec{x})\\
(\frac{1}{2}i\gamma^0\gamma^3\gamma^5+x^1\partial_2-x^2\partial_1)\ 
\mathcal{H}_M^{-1}\{\Psi\}(\vec{x})&=(\mathcal{H}_M^{-1}\circ
R')\{\Psi\}(\vec{x})
\end{align*}
Where  $\Psi\in Pinor(\mathbb{S}^3)$  and $R,R':Pinor(\mathbb{S}^3)\to Pinor(\mathbb{S}^3)$ are
linear maps defined by:
\begin{align*}
R\{\Psi\}(p,l,\mu)&\equiv i\gamma^0 E_p \Psi(p,l,\mu)\\
R'\{\Psi\}(p,l,\mu)&\equiv i\gamma^0 (\mu+\frac{1}{2})\Psi(p,l,\mu)
\end{align*}
\end{prop}

\begin{proof}
We have $\mathcal{H}^{-1}_{M}=(\mathcal{H}^{\Theta}_{P})^{-1}\circ
S^{-1}$. Then we can check that
$i\gamma^5\Lambda_{l\mu}(pr,\theta,\varphi)=-(-1)^\mu\Lambda_{l,-\mu-1}(pr,\theta,\varphi)i\gamma^1$.

Therefore, the inverse Pauli-Hankel Transform verifies:
\begin{align*}
(\gamma^0\vec{\gamma}\cdot \vec{\partial}+i\gamma^0 m)\ (\mathcal{H}^\Theta_P)^{-1}\{\Psi\}(\vec{x})=((\mathcal{H}^\Theta_P)^{-1}\circ Q)\{\psi\}(\vec{x})
\end{align*}
Where  $\Psi\in Pinor(\mathbb{S}^3)$  and $Q:Pinor(\mathbb{S}^3)\to Pinor(\mathbb{S}^3)$ is a
linear map defined by:
\begin{align*}
\left[ \begin{array}{l}
Q\{\Psi\}(p,l,\mu)\\
Q\{\Psi\}(p,l,-\mu-1)
\end{array}
\right]&\equiv
\left[ \begin{array}{cc}
i\gamma^0 m
& (-1)^\mu \gamma^0\gamma^3 p\\
-(-1)^\mu \gamma^0\gamma^3 p 
&i\gamma^0 m
\end{array}
\right]\ 
\left[ \begin{array}{l}
\Psi(p,l,\mu)\\
\Psi(p,l,-\mu-1)
\end{array}
\right]
\end{align*}
Now we show that $Q\circ S^{-1}=S^{-1}\circ R$:
\begin{align*}
&\left[ \begin{array}{cc}
i\gamma^0 m
& (-1)^\mu \gamma^0\gamma^3 p\\
-(-1)^\mu \gamma^0\gamma^3 p 
&i\gamma^0 m
\end{array}
\right]\ 
\left[ \begin{array}{cc}
\sqrt{\frac{E_p+m}{2E_p}} 
& \sqrt{\frac{E_p-m}{2E_p}}(-1)^\mu i\gamma^3\\
-\sqrt{\frac{E_p-m}{2E_p}}(-1)^\mu i\gamma^3 
&\sqrt{\frac{E_p+m}{2E_p}}
\end{array}
\right]=\\
&=\left[ \begin{array}{cc}
\sqrt{\frac{E_p+m}{2E_p}} 
& \sqrt{\frac{E_p-m}{2E_p}}(-1)^\mu i\gamma^3\\
-\sqrt{\frac{E_p-m}{2E_p}}(-1)^\mu i\gamma^3 
&\sqrt{\frac{E_p+m}{2E_p}}
\end{array}\right]\ 
 \left[ \begin{array}{cc}
 i\gamma^0 E_p & 0\\
 0 & i\gamma^0 E_p
 \end{array}
 \right]
\end{align*}

The inverse Pauli-Hankel Transform also verifies:
\begin{align*}
(\frac{1}{2}i\gamma^0\gamma^3\gamma^5+x^1\partial_2-x^2\partial_1)\ (\mathcal{H}^\Theta_P)^{-1}\{\Psi\}(\vec{x})=((\mathcal{H}^\Theta_P)^{-1}\circ Q')\{\psi\}(\vec{x})
\end{align*}
Where  $\Psi\in Pinor(\mathbb{S}^3)$  and $R':Pinor(\mathbb{S}^3)\to Pinor(\mathbb{S}^3)$ is the
 linear map defined by:
\begin{align*}
\left[ \begin{array}{l}
R'\{\Psi\}(p,l,\mu)\\
R'\{\Psi\}(p,l,-\mu-1)
\end{array}
\right]&\equiv
\left[ \begin{array}{cc}
i\gamma^0(\mu+\frac{1}{2})
& 0\\
0
&-i\gamma^0(\mu+\frac{1}{2})
\end{array}
\right]\ 
\left[ \begin{array}{l}
\Psi(p,l,\mu)\\
\Psi(p,l,-\mu-1)
\end{array}
\right]
\end{align*}
It verifies $R'\circ S^{-1}=S^{-1}\circ R'$.
\end{proof}

\section{Majorana spinor field representation of the Poincare group}

Consider a Majorana spinor field $\Psi\in Pinor(\mathbb{R}^3)$.
Let the Dirac Hamiltonian, $H$, be defined in the configuration space by:
\begin{align*}
iH\{\Psi\}(\vec{x})\equiv (\gamma^0\vec{\gamma}\cdot \vec{\partial}+i\gamma^0
m)\Psi(\vec{x}),\ m\geq 0
\end{align*}
In the momentum space:
\begin{align*}
iH\{\Psi\}(\vec{p})\equiv i\gamma^0E_p \Psi(\vec{p})
\end{align*}
The free Dirac equation is verified by:
\begin{align*}
(\partial_0+iH) e^{-iH x^0}\{\Psi\}=0
\end{align*}

\begin{defn}
Given a Majorana spinor field 
$\Psi\in Pinor(\mathbb{R}^3)$, we define $\Psi(x)\equiv e^{-iHx^0}\{\Psi\}(\vec{x})$.
The Majorana spinor field projective representation of the Poincare
group is defined, up to a sign, as:
\begin{align*}
P(\Lambda_S,b)\{\Psi\}(x)\equiv \pm S\Psi(\Lambda^{-1}_S x+b)
\end{align*}
Where $\Lambda_S\in O(1,3)$, $S\in Pin(3,1)$ is such that 
$\Lambda^\mu_{S\ \nu}\gamma^\nu=S\gamma^\mu S^{-1}$ and $b\in\mathbb{R}^4$.
\end{defn}

\begin{prop}
The Majorana spinor field representation of the inhomogeneous
restricted Lorentz group, for a finite mass, is irreducible and unitary.
\end{prop}
\begin{proof}
Suppose that we have for some $\Phi$ and $\Psi$, that for all $a\in \mathbb{R}^4$:
\begin{align*}
<\Phi, P(1,a)\{\Psi\}>=0 
\end{align*}
Doing a Fourier transform, the above equation can be written as:
\begin{align*}
\int \frac{d^3\vec{p}}{(2\pi)^3}\ \Phi^{\dagger}(\vec{p})e^{-i\gamma^0
  p \cdot a}\Psi(\vec{p})=0
\end{align*}

\begin{align*}
\int \frac{d^3\vec{p}}{(2\pi)^3}\
\Phi^{\dagger}(\vec{p})\Big(\frac{1+\gamma^0}{2}e^{-i p \cdot
  a}+\frac{1-\gamma^0}{2}e^{i p \cdot
  a})\Psi(\vec{p})=0
\end{align*}

Now we multiply it by $e^{-i \vec{q}\cdot \vec{a}}$, with $\vec{q}$ arbitrary. Integrating in $\vec{a}$, we get:

\begin{align*}
\Phi^{\dagger}(\vec{q})\frac{1+\gamma^0}{2}\Psi(\vec{q})e^{-i E_q a^0}+\Phi^{\dagger}(-\vec{q})\frac{1-\gamma^0}{2}\Psi(-\vec{q})e^{i E_q a^0}=0
\end{align*}
If we multiply the equation above by $e^{iE_q a^0}$ and we integrate
$a^0$ from $0$ to $2\pi/E_q$, we get $\Phi^{\dagger}(\vec{q})\frac{1+\gamma^0}{2}\Psi(\vec{q})=0$. Considering
real and imaginary parts in separate, we obtain 
$\Phi^{\dagger}(\vec{q})\Psi(\vec{q})=0$ and
$\Phi^{\dagger}(\vec{q})i\gamma^0\Psi(\vec{q})=0$.

Suppose $S\in Spin^+(3,1)$ verifies $S\slashed q=\slashed p S$. 
Then it can be written as:
$S=B_p R B^{-1}_q$, where $R \slashed l=\slashed l R$ and $B_p$ is any
Lorentz transform verifying $B_p\slashed l=\slashed p B_p$.
Now suppose $i\slashed l=i\gamma^0 m$, with $m>0$. Then 
$B_p\equiv \frac{\slashed p
\gamma^0+m}{\sqrt{E_p+m}\sqrt{2m}}$, where $p^0=E_p$, satisfies $B_p\slashed l=\slashed
p B_p$. 
Then $R$ is a representation of $SU(2)$ and:
\begin{align*}
&S\{\Psi\}(x)=\int \frac{d^3\vec{p}}{\sqrt{(2\pi)^3}}S\frac{\slashed p
  \gamma^0+m}{\sqrt{E_p+m}\sqrt{2E_p}}e^{-i\gamma^0 \Lambda(p) \cdot
x}\Psi(\vec{p})\\
&=\int \frac{d^3\vec{p}}{\sqrt{(2\pi)^3}}\frac{\slashed \Lambda (p)
  \gamma^0+m}{\sqrt{\Lambda^0(p)+m}\sqrt{2\Lambda^0(p)}}e^{-i\gamma^0 \Lambda(p) \cdot
x}R\sqrt{\frac{\Lambda^0(p)}{E_p}}\Psi(\vec{p})\\
&=\int \frac{d^3\vec{p}}{\sqrt{(2\pi)^3}}\frac{(\Lambda^{-1})^0(p)}{E_p}\frac{\slashed p
  \gamma^0+m}{\sqrt{E_p+m}\sqrt{2E_p}}e^{-i\gamma^0 p \cdot
x}R\sqrt{\frac{E_p}{(\Lambda^{-1})^0(p)}}\Psi(\vec{\Lambda}^{-1}(p))
\end{align*}
Then:
\begin{align*}
\mathcal{F}_M\circ S\{\Psi\}(x^0,\vec{p})=e^{-i\gamma^0 E_p x^0}R\sqrt{\frac{(\Lambda^{-1})^0(p)}{E_p}}\Psi(\vec{\Lambda}^{-1}(p))
\end{align*}

Hence the Poincare representation is unitary.
Since $m>0$, for all $\vec{q}$ and $\vec{p}$, we can always find $\Lambda$ such
that $\vec{q}=\vec{\Lambda}(p)$. If the Poincare representation is reducible, since it is unitary,
there are 2 states $\Psi,\Phi$ verifying for all $g\in SL(2,\mathbb{C})$ and
$a\in \mathbb{R}^4$:
\begin{align*}
<\Phi,S_g\circ T(a)\{\Psi\}>=0 
\end{align*}

This implies that for all $\vec{p}$ and $\vec{q}$:
\begin{align*}
\frac{m}{E_p}\Phi^\dagger(\vec{q})R\Psi(\vec{p})=0
\end{align*}
$R$ is a Majorana representation of $SU(2)$, which from Proposition
\ref{prop:irreducible} is irreducible,
so the equation above is not true. Therefore the Poincare representation is irreducible and unitary.
\end{proof}

The translations in space-time are given by $P(1,b)$. Doing a Fourier-Majorana
transform, we get: $P(1,b)\{\Psi\}(x^0,\vec{p})\equiv e^{-i\gamma^0
  p\cdot b}\Psi(x^0,\vec{p})$, with $p^2=m^2$. 
Therefore, $p$ is related with the 4-momentum of the Poincare
representation.

The rotations are defined by $P(R,0)$, where $R\in SU(2)$. Doing a Hankel-Majorana
transform, we get for a rotation along $z$ by an angle $\theta$: 
\begin{align*}
P(R,0)\{\Psi\}(x^0,p,l,\mu)\equiv
e^{i\gamma^0(\mu+\frac{1}{2})\theta}\Psi(x^0,p,l,\mu)
\end{align*}
  
Therefore, $\mu$ is related with the angular momentum of a spin
one-half Poincare representation.

Additionally, the transition operator $T$ defined by:
\begin{align*}
\Psi(x)&=\int d^3\vec{y} T(x-y)\Psi(y)
\end{align*}
It is given by:
\begin{align*}
T(x)= \int \frac{d^3\vec{p}}{(2\pi)^3}\frac{\slashed p
  \gamma^0+m}{\sqrt{E_p+m}\sqrt{2E_p}}e^{-i\gamma^0 p \cdot
x}\frac{\slashed p\gamma^0+m}{\sqrt{E_p+m}\sqrt{2E_p}}
\end{align*}
When $x^0=0$ and $\vec{x}\neq 0$, $T(x)=0$. Doing a
Lorentz transformation we get that when $x^2<0$, $T(x)=0$ and
therefore the transition operator respects relativistic causality in
the sense that it is null outside the light cone.

\section{Energy of Majorana spinor fields}

\begin{defn}
The Energy Transform $\mathcal{E}: Pinor(\mathbb{R})\to
Pinor(\mathbb{R})$ is an operator defined by:
\begin{align*}
\mathcal{E}\{\Psi\}(p^0)&\equiv \int dx^0\ \frac{e^{i\gamma^0p^0x^0}}{\sqrt{2\pi}}\Psi(x^0),\ \Psi\in Pinor(\mathbb{R})
\end{align*}
Where the domain of the integral is $\mathbb{R}$, $m\geq 0$.
\end{defn}

\begin{prop}
The Energy transform is an unitary operator.
\end{prop}
\begin{proof}
The Energy transform can be written as:
\begin{align*}
\mathcal{E}\{\Psi\}(p^0)=\Theta_{L^2}\circ
\mathcal{F}_P(-p^0)\circ\Theta^{-1}_{L^2}\{\Psi\}
\end{align*}
Where $\mathcal{F}_P(-p^0)$ is a Pauli-Fourier transform over
$\mathbb{R}$ and $\Theta$ was defined in Definition \ref{defn:Theta}.
Since the Pauli-Fourier transform is unitary, so is the
Energy transform.
\end{proof}

The energy transform can be applied in the time coordinate of a
Majorana spinor field, $x^0$, after a (linear or spherical)
momentum transform on the space coordinates, $\vec{x}$, to define an
unitary energy-momentum transform:\\
- for the linear case $\mathcal{E}\circ
\mathcal{F}_M:Pinor(\mathbb{R}^4)\to Pinor(\mathbb{R}^4)$;\\
- for the spherical case $\mathcal{E}\circ
\mathcal{H}_M:Pinor(\mathbb{R}^4)\to Pinor(\mathbb{R}\times
\mathbb{S}^3)$.

\section{Conclusion}

There are Poincare group representations
on complex Hilbert spaces, like the Dirac spinor field, 
or real Hilbert spaces, like the electromagnetic field
tensor. Therefore, the study of the Poincare group representations
should be independent on whether the representations are defined on
real or complex Hilbert spaces.

We showed that the Majorana spinor field with finite mass is an
unitary irreducible spin one-half representation 
of the Poincare group on a real Hilbert space.

Since the Bargmann-Wigner
equations are valid for all spins and are based on
the free Dirac equation, these results
open the possibility to study Poincare group representations with
arbitrary spins on real Hilbert spaces.

\section*{Acknowledgements}
Leonardo Pedro's work was supported by FCT under contract
SFRH/BD/70688/2010.
Leonardo thanks R. M. Fonseca, J. Mourao, J. Natario for the important suggestions and corrections.
Leonardo thanks  A. Carmona, D. Emmanuel-Costa, A. Moita, P.B. Pal,
N. Ribeiro, J.C. Romao, J.I. Silva-Marcos for the helpful discussions.
\addcontentsline{toc}{section}{References}
\bibliography{qed}{}
\bibliographystyle{elsarticle-num}
\end{document}